\documentclass[11pt]{article}
\usepackage{amssymb,amsmath,amsthm,graphicx}
\usepackage{subcaption}
\usepackage{float}
\usepackage{color,enumerate}
\usepackage{fullpage}
\usepackage{hyperref}
\usepackage{cite}
\usepackage[ruled,vlined,linesnumbered]{algorithm2e}
\usepackage{xspace}
\usepackage{color}

    \theoremstyle{plain}
  \newtheorem{definition}{Definition}
 \newtheorem{theorem}{Theorem}

  \newtheorem{corollary}[theorem]{Corollary}
 
  \newtheorem{lemma}[theorem]{Lemma}

    \theoremstyle{remark}

\newcommand{\card}[1]{\vert{}#1\vert{}}

\newcommand{\remove}[1]{}

\newcommand{\IDS}{\textsc{Identifiable Subgraph}\xspace}
\newcommand{\MAXIDS}{\textsc{Max-Identifiable Subgraph}\xspace}
\newcommand{\MINIDS}{\textsc{Min-Identifiable Subgraph}\xspace}
\newcommand{\MINIDSk}{\textsc{Min-Identifiable Subgraph}($k$)\xspace}
\newcommand{\MINIDSnk}{\textsc{Min-Identifiable Subgraph}($|L|-k$)\xspace}

\newcommand{\MCCk}{\textsc{Multicolored Clique}($k$)\xspace}

\newcommand{\NP}{\ensuremath{\mathsf{NP}}\xspace}
\newcommand{\APX}{\ensuremath{\mathsf{APX}}\xspace}
\newcommand{\FPT}{\ensuremath{\mathsf{FPT}}\xspace}
\newcommand{\W}[1]{\ensuremath{\mathsf{W[#1]}}\xspace}

\newcommand{\N}{\mathbb{N}}

\title{On the complexity of the identifiable subgraph problem, revisited}

\author{Stefan Kratsch\\
\small University of Bonn, Institute of Computer Science, Friedrich-Ebert-Allee 144, D-53113 Bonn, Germany\\
\small \texttt{kratsch@cs.uni-bonn.de}
\and
Martin Milani\v c\\
\small University of Primorska, UP IAM, Muzejski trg 2, SI-6000 Koper, Slovenia\\
\small University of Primorska, UP FAMNIT, Glagolja\v ska 8, SI-6000 Koper, Slovenia\\
\small \texttt{martin.milanic@upr.si}
}

\begin{document}
\maketitle

\begin{abstract}
A bipartite graph $G=(L,R;E)$ with at least one edge is said to be \emph{identifiable} if for every vertex $v\in L$, the subgraph induced by its non-neighbors has a matching of cardinality $\card{L}-1$. An \hbox{\emph{$\ell$-subgraph}} of $G$ is an induced subgraph of $G$ obtained by deleting from it some vertices in $L$ together with all their neighbors. The {\sc Identifiable Subgraph} problem is the problem of determining whether a given
bipartite graph contains an identifiable $\ell$-subgraph.

We show that the {\sc Identifiable Subgraph} problem is polynomially solvable, along with the version of the problem in which the task is to delete as few vertices from $L$ as possible together with all their neighbors so that the resulting $\ell$-subgraph is identifiable.
We also complement a known \APX-hardness result for the complementary problem in which the task is to minimize the number
of remaining vertices in $L$, by showing that two parameterized variants of the problem are \W{1}-hard.
\end{abstract}

%---------------------------------------------------------------------------------------------------------------------------------

\section{Introduction} \label{sec:Introduction}

A \emph{matching} in a graph is a subset of pairwise disjoint edges.  A bipartite graph $G=(L,R;E)$ with at least one edge is said to be \emph{identifiable} if for every vertex in~$L$, the subgraph of $G$ induced by its non-neighborhood has a matching of cardinality~$\card{L}-1$. Identifiable bipartite graphs were studied in several papers~\cite{DAM1,ALGO,COC,DAM2}; the property arises in the context of low-rank matrix factorization and has applications in data mining, signal processing, and computational biology. For further details on applications of notions and problems discussed in this paper, we refer to~\cite{DAM1,ALGO}.

While the recognition problem for identifiable bipartite graphs is clearly polynomial using bipartite matching algorithms, several natural algorithmic problems concerning identifiable graphs turn out to be {\sf NP}-complete (see~\cite{DAM1,ALGO,DAM2}). In~\cite{DAM1}, three problems related to finding specific identifiable subgraphs were introduced. To state these problems, we need to recall the notion of an $\ell$-subgraph of a bipartite graph (which appeared first in~\cite{DAM1} and, in a slightly modified form, which we will adopt, in~\cite{DAM2}).
For a bipartite graph $G=(L,R;E)$ and vertex sets $X \subseteq L$, $Y \subseteq R$, we denote by  $G[X,Y]$ the subgraph of $G$ induced by $X \cup Y$.

\begin{definition}
Let $G=(L,R;E)$ be a bipartite graph. For a subset $J\subseteq L$, the \emph{$\ell$-subgraph of $G$ induced by $J$} is the subgraph $G(J) = G[J, R \setminus N(L\setminus J)]$, where $N(L\setminus J)$ denotes the set of all vertices in $R$ with a neighbor in $L\setminus J$. We say that a graph $G'$ is an \emph{$\ell$-subgraph of $G$} if there exists a subset $J\subseteq L$ such that $G'=G(J)$.
\end{definition}

The following three problems are all related to finding identifiable $\ell$-subgraphs of a given graph:

\medskip
\begin{center}
\fbox{\parbox{0.89\linewidth}{\noindent
{\sc Identifiable Subgraph}
\\[.8ex]
\begin{tabular*}{.95\textwidth}{rl}
\emph{Instance:} & A bipartite graph $G = (L,R;E)$.\\
\emph{Question:} & Does $G$ have an identifiable $\ell$-subgraph?
\end{tabular*}
}}
\end{center}

\begin{center}
\fbox{\parbox{0.89\linewidth}{\noindent
{\sc Min-Identifiable Subgraph}
\\[.8ex]
\begin{tabular*}{.95\textwidth}{rl}
\emph{Instance:} & A bipartite graph $G = (L,R;E)$ and an integer $k$.\\
\emph{Question:} & Does $G$ have an identifiable $\ell$-subgraph induced by a set $J$ with $|J|\le k$?
\end{tabular*}
}}
\end{center}

\begin{center}
\fbox{\parbox{0.89\linewidth}{\noindent
{\sc Max-Identifiable Subgraph}
\\[.8ex]
\begin{tabular*}{.95\textwidth}{rl}
\emph{Instance:} & A bipartite graph $G = (L,R;E)$ and an integer $k$.\\
\emph{Question:} & Does $G$ have an identifiable $\ell$-subgraph induced by a set $J$ with $|J|\ge k$?
\end{tabular*}
}}
\end{center}

\medskip
In~\cite{DAM1}, the optimization version of the \MINIDS problem was shown to be \APX-hard.
In the same paper it was shown that all three problems are polynomially solvable for trees, as well as for bipartite graphs $G = (L,R;E)$ such that the maximum degree of vertices in $L$ is at most~$2$. In~\cite{DAM2}, restricted versions of the \IDS problem were studied, parameterizing the instances
according to the maximum degree $\Delta(R)$ of vertices in $R$.
Formally:

\begin{center}
\fbox{\parbox{0.89\linewidth}{\noindent
{\sc $k$-bounded Identifiable Subgraph}
\\[.8ex]
\begin{tabular*}{.95\textwidth}{rl}
\emph{Instance:} & A bipartite graph $G = (L,R;E)$ with $\Delta(R)\le k$.\\
\emph{Question:} & Does $G$ have an identifiable $\ell$-subgraph?
\end{tabular*}
}}
\end{center}
\medskip
\noindent

It was shown in~\cite{DAM2} that the {\sc $k$-bounded Identifiable Subgraph} problem for $k\ge 3$ is as hard as \IDS problem in general and that the
{\sc $2$-bounded Identifiable Subgraph} problem is solvable in linear time.
The complexity of the \IDS and \MAXIDS problems in general bipartite graphs was left open by previous works.

\medskip
In this paper, we establish the computational complexity of the \IDS and \MAXIDS problems, showing that both problems are solvable in polynomial time.
The key idea to our approach is the observation that if the input graph $G = (L,R;E)$ is not identifiable, then one can
compute in polynomial time a maximal subset $K\subseteq L$ no vertex of which is contained in any identifiable $\ell$-subgraph of $G$.
Such a set $K$ is non-empty and can be safely deleted from the graph together with all its neighbors, thus reducing the problem to a smaller
graph. If the algorithm finds an identifiable $\ell$-subgraph of $G$, then it in fact finds an
identifiable $\ell$-subgraph of $G$ induced by a largest possible subset of $L$, thereby also solving the
\MAXIDS problem. The proof also shows that such a subgraph is unique.

In the second part of the paper, we complement the \APX-hardness result for the optimization version of the \MINIDS problem from~\cite{DAM1} by
studying the problem from the parameterized complexity point of view. We introduce two natural parameterized variants of the \MINIDS problem and prove that both are \W{1}-hard, by giving parameterized reductions from the well-known \W{1}-hard \MCCk problem.

The paper is structured as follows. In Section~\ref{sec:prelim}, we give the necessary definitions.
In Section~\ref{sec:algorithm}, we give a polynomial time algorithm that simultaneously solves the \IDS and the \MAXIDS problems.
In Section~\ref{sec:parameterized}, we study the \NP-hard \MINIDS problem from the parameterized complexity point of view.
Section~\ref{sec:conclusion} concludes the paper with some open questions.

\section{Preliminaries}\label{sec:prelim}

All graphs considered in this paper are finite, simple, and undirected. For a graph $G$, we denote by $V(G)$ the vertex set of $G$ and by $E(G)$ its edge set.
A \emph{bipartite graph} is a graph $G = (V,E)$ such that there exists a partition of $V$ into two sets $L$ and $R$ such that $L\cap R=\emptyset$ and $E\subseteq \{\{\ell, r\}~;\ell\in L \textrm{~and~} r\in R\}$. In this paper, we will regard bipartite graphs as already \emph{bipartitioned}, that is, given together with a fixed bipartition $(L,R)$ of their vertex set, and hence use the notation $G = (L,R;E)$. For a graph $G=(V,E)$ and a subset of vertices $X\subseteq V$, $N_G(X)$ denotes the neighborhood of $X$, i.e., the set of all vertices in $V \setminus X$ that have a neighbor in $X$. For a vertex $x\in V$, we write $N_G(x)$ for $N_G(\{x\})$, and denote the \emph{degree} of $x$ with $d_G(x)=|N_G(x)|$. In $N_G(X)$, $N_G(x)$, $d_G(x)$, we shall omit the subscript $G$ if the graph is clear from the context.
A {\it clique} in a graph is a set of pairwise adjacent vertices.

A \emph{parameterized problem} is a language $Q\subseteq\Sigma^*\times\N$; the second component, $k$, of instances $(x,k)\in\Sigma^*\times\N$ is called the \emph{parameter}. A parameterized problem $Q$ is \emph{fixed-parameter tractable} (FPT) if there is a function $f\colon\N\to\N$, a constant $c$, and an algorithm $A$ that decides $(x,k)\in Q$ in time $f(k)|x|^c$ for all $(x,k)\in\Sigma^*\times\N$. Let $\FPT$ denote the class of all fixed-parameter tractable parameterized problems. A \emph{parameterized reduction} from $Q\subseteq\Sigma^*\times\N$ to $Q'\subseteq\Sigma'^*\times\N$ is a mapping $\pi\colon\Sigma^*\times\N\to\Sigma'^*\times\N$ such that there are functions $g,h\colon\N\to\N$ and a constant $c$ with: $(x,k)\in Q$ if and only if $\pi((x,k))\in Q'$, the parameter value $k'$ of $(x',k')=\pi((x,k))$ is at most $g(k)$, and $\pi((x,k))$ can be computed in time $h(k)|x|^c$. It is well known that the existence of a parameterized reduction from $Q$ to $Q'$ and $Q'\in\FPT$ imply that $Q\in\FPT$ as well, and that parameterized reducibility is transitive. Accordingly, similarly to $\mathsf{P}$ vs. $\NP$, there are hardness classes of problems that are suspected not to be FPT. In particular, it is believed that $\FPT\neq\W{1}$ and, under this assumption, a parameterized reduction from any \W{1}-hard problem rules out fixed-parameter tractability. (Here \W{1}-hardness is with respect to parameterized reductions.)

For graph-theoretic definitions not given in the paper we refer to~\cite{MR2744811,MR1367739},
for further background in matching theory to~\cite{MR859549}, and for background in
parameterized complexity to~\cite{MR3380745,DowneyF13}.
%%%%%%%%%%%%%%%%%%%

\section{A polynomial time algorithm for the \IDS and the \MAXIDS problems} \label{sec:algorithm}

In this section we give a polynomial time algorithm for the \IDS problem, the problem of determining whether a given graph $G=(L,R;E)$ has an identifiable $\ell$-subgraph. As a corollary of our approach we will also obtain a polynomial time algorithm for the \MAXIDS problem.

The key ingredient for the algorithm is the following lemma.

\begin{lemma}\label{lemma:ptime:key}
Let $G=(L,R;E)$ be a non-identifiable bipartite graph with at least one edge and let $v\in L$ such that
there is no matching of $L\setminus \{v\}$ into $R\setminus N(v)$.
Let $K$ be an inclusion-wise minimal subset of $L\setminus \{v\}$ that has no matching into $R\setminus N(v)$.
Such a set $K$ is nonempty and always exists.
Moreover, no identifiable $\ell$-subgraph of $G$ contains a vertex of $K$.
\end{lemma}

\begin{proof}
Let $K$ be a minimal subset of $L\setminus \{v\}$ that has no matching into $R \setminus N(v)$. By Hall's Theorem there must be a subset $K'\subseteq K$ with $|N(K')\cap (R\setminus N(v))|<|K'|$. Any such set $K'$ has no matching into $R\setminus N(v)$. Because $K$ is a minimal set without a matching into $R\setminus N(v)$ it follows that $|N(K)\cap (R\setminus N(v))|<|K|$.
Furthermore, every proper subset of $K$ does have a matching into $R\setminus N(v)$.

Now, fix an arbitrary set $J\subseteq L$ such that the induced $\ell$-subgraph $G'=G[J,R\setminus N(L\setminus J)]$ is identifiable.
We need to show that $J\cap K=\emptyset$.

Assume for contradiction that $K\cap J\neq \emptyset$. Let $K_{\it in}=K\cap J$ and $K_{\it out}=K\setminus J$. Because $K_{\it in}\neq\emptyset$ we have that $K_{\it out}$ is a proper, possibly empty, subset of $K$. Hence, by the first paragraph, we have that $|K_{\it out}|\leq |N(K_{\it out})\cap (R\setminus N(v))|$. In the $\ell$-subgraph $G_J$ induced by $J$, by definition, none of the neighbors of $K_{\it out}$ are present. Thus, the vertices in $K_{\it in}$ have at most those vertices as neighbors that are adjacent to $K_{\it in}$ but not to $K_{\it out}$. (Further vertices in $L\setminus (J\cup K)$ may imply that further neighbors of $K_{\it in}$ are not present, but this will not be important.)
Thus, the number of neighbors that $K_{\it in}$ has in the vertices of $R\setminus N(v)$ that are present in $G_J$ is at most
\[
|N_{G_J}(K)|-|N_{G_J}(K_{\it out})|< |K| - |K_{\it out}|=|K_{\it in}|.
\]
It follows immediately that $K_{\it in}$ has no matching into $R\setminus N(v)$ in $G_J$. If $v\in J$ then testing the identifiability condition for $v$ would require such a matching. If $v\notin J$ then using that $G_J$ must have a matching of $J$ into $N_{G_J}(J)$
means that we would need a matching of $K_{\it in}$ into $N_{G_J}(K_{\it in})\subseteq R\setminus N(v)$. Thus, either way we get a contraction. This implies that $J\cap K=\emptyset$, as claimed.
\end{proof}

Given a graph $G$ and vertex $v$ as in Lemma \ref{lemma:ptime:key} the set $K$ can be found in a straightforward way by folklore knowledge about bipartite matchings. We sketch a very simple algorithm by self-reduction for completeness.

\begin{lemma}\label{lemma:ptime:minimalset}
Given a non-identifiable graph $G=(L,R;E)$ and vertex $v\in L$
such that there is no matching of $L\setminus \{v\}$ into $R\setminus N(v)$,
a minimal set $K$ as in Lemma \ref{lemma:ptime:key} can be found in polynomial time.
\end{lemma}

\begin{proof}
Set $K:=L\setminus \{v\}$ and repeat the following routine: Try each vertex $w\in K$ and test whether there is a
matching of $K\setminus \{w\}$ into $R\setminus N(v)$. If there is then try the next vertex. If not then update $K:=K\setminus \{w\}$ and repeat. Output the current set $K$ if each $K\setminus \{w\}$ has a matching of $K\setminus \{w\}$ into $R\setminus N(v)$.

As an invariant, the set $K$ never has a matching into $R\setminus N(v)$. In particular, we can never reach an empty set (and we can only reach a singleton vertex if it is isolated). Thus, the algorithm must terminate with a nonempty set $K$ such that each set $K\setminus \{w\}$ has a matching into $R\setminus N(v)$. This also means that all smaller subsets of $K$ have matchings into $R\setminus N(v)$. Thus, $K$ is a minimal set with no matching into $R\setminus N(v)$.
\end{proof}

We know now that if $G=(L,R;E)$ is not identifiable then we can efficiently find a subset $K\subseteq L$ such that no vertex of $K$ is contained in any identifiable $\ell$-subgraph of $G$. We now prove formally that we may safely delete $K$ and $N(K)$ from $G$ while still retaining the same set of identifiable $\ell$-subgraphs.

\begin{lemma}\label{lemma:lsubgraph:deletion}
Let $G=(L,R;E)$ a bipartite graph and let $K\subseteq L$. Every $\ell$-subgraph of $G$ that contains no vertex of $K$ is also an $\ell$-subgraph of
the $\ell$-subgraph of $G$ induced by $L\setminus K$,
and vice versa.
Moreover, these $\ell$-subgraphs are induced by the same sets $J\subseteq L\setminus K$.
\end{lemma}

\begin{proof}
Every $\ell$-subgraph of a graph is defined by the left part of its bipartition. We show that taking the induced $\ell$-subgraph for any $J\subseteq L\setminus K$ gives the same graph from both $G$ and $G-N[K]=G[L\setminus K,R\setminus N(K)]$. Fix an arbitrary set $J\subseteq L\setminus K$.

Clearly, since $J$ is a subset of the left part of the bipartition in both graphs, we get $\ell$-subgraphs of the form $H_1=G[J,R_1]$ and $H_2=(G-N[K])[J,R_2]$. The latter is also an induced subgraph of $G$ so it simplifies to $H_2=G[J,R_2]$. It suffices to prove that $R_1=R_2$.

By definition of $\ell$-subgraph we have $R_1=R\setminus N_G(L\setminus J)$.
Similarly, for the $\ell$-subgraph of $J$ in $G'=G-N_G[K]=G[L\setminus K,R\setminus N_G(K)]$ we get
\[
R_2=(R\setminus N_G(K))\setminus N_{G'}((L\setminus K)\setminus J)\,.
\]
We can safely replace $N_{G'}((L\setminus K)\setminus J)$ by $N_G((L\setminus K)\setminus J)$ because $G'$ is an induced subgraph of $G$ so the neighborhood is only affected by restriction to $R\setminus N_G(K)$, the right part of the bipartition of $G'$. Thus, in $R_2$ we have the vertices of $R$ that do not have a neighbor in $K$ and that do not have a neighbor in $(L\setminus K)\setminus J$. Because $K$ and $J$ are disjoint subsets of $L$ this is the same as taking out the neighbors of $L\setminus J$ from $R$, i.e., taking $R\setminus N(L\setminus J)=R_1$. Thus, both graphs are induced subgraphs of $G$ with left part $J$ and right part $R\setminus N(L\setminus J)$, so they are identical as claimed.
\end{proof}

In particular, the lemma implies that if no identifiable $\ell$-subgraph contains a vertex of a nonempty set $K\subseteq L$ then $G$ and $G-N[K]$ contain the same identifiable $\ell$-subgraphs. Thus, when seeking identifiable $\ell$-subgraphs it is safe to eliminate $N[K]$ for sets $K$ obtained via Lemma~\ref{lemma:ptime:key}.

Now we can put together the claimed polynomial time algorithm.

\begin{theorem}\label{theorem:ids:ptime}
The \IDS problem can be solved in polynomial time.
\end{theorem}

\begin{proof}
The algorithm works as follows. Given an input graph $G=(L,R;E)$ it
first tests if $E = \emptyset$.
If $E = \emptyset$, then $G$ is not identifiable and has no identifiable $\ell$-subgraph; the algorithm reports this fact and halts.
If $E\neq\emptyset$, the algorithm proceeds iteratively.
Identifiability can be efficiently tested by $|L|$ bipartite matching computations.
If $G$ is identifiable then graph $G$ is output as an identifiable $\ell$-subgraph.
If $G$ is not identifiable then the algorithm picks an arbitrary $v$ such that there is no matching of $L\setminus \{v\}$ into $R\setminus N(v)$. By Lemma \ref{lemma:ptime:key} there is a nonempty set $K\subseteq L$ such that no identifiable $\ell$-subgraph of $G$ contains a vertex of $K$; such a set can be found efficiently by Lemma~\ref{lemma:ptime:minimalset}. Thus, if $G$ has any identifiable $\ell$-subgraph then every such subgraph must avoid $K$ and, hence, it is also an $\ell$-subgraph of $G-N[K]$ by Lemma~\ref{lemma:lsubgraph:deletion}. Conversely, $G-N[K]$ contains no further identifiable $\ell$-subgraphs. The algorithm thus replaces $G$ by $G-N[K]$ and starts over. In case a graph is output, Lemma~\ref{lemma:lsubgraph:deletion} implies that the output graph is also an $\ell$-subgraph of the initial graph $G$.
\end{proof}

In fact, it can be easily seen that the algorithm always returns a maximum identifiable $\ell$-subgraph and thus also solves the maximization variant
of the problem, \MAXIDS, in polynomial time. (The proof also shows that this graph is unique.)

\begin{corollary}\label{corollary:max-ids:ptime}
The \MAXIDS problem can be solved in polynomial time.
\end{corollary}

\begin{proof}
Clearly, if the input graph is identifiable then returning it is optimal.
If not then
either $E = \emptyset$ (in which case $G$ has no identifiable $\ell$-subgraph) or
% MM: added the line above to handle edgeless graphs
there is a vertex $v$ such that $L\setminus \{v\}$ cannot be matched into $R\setminus N(v)$ and, by Lemma \ref{lemma:ptime:key} the algorithm finds a nonempty
% MM: added "nonempty"
set $K$ that is avoided by all identifiable $\ell$-subgraphs. Since $G$ and $G-N[K]$ have the same $\ell$-subgraphs induced by $J\subseteq L\setminus K$, in particular, any maximum identifiable $\ell$-subgraph of $G$ is also an identifiable $\ell$-subgraph of $G-N[K]$. Thus, continuing the iterative approach on $G-N[K]$ will find a maximum solution, if one exists.
\end{proof}

%%%%%%%%%%%%%%%%%%%%%%%%%%%%%%%%%%%%%%%%%%%%%%%%%%%%%%%%%%%%%%%

\section{Parameterized complexity of Min-Identifiable Subgraph}\label{sec:parameterized}

In this section we study the parameterized complexity of the \MINIDS problem, which was proved \NP-hard in a previous work~\cite{DAM1}. We consider the following parameterized variants \MINIDSk and \MINIDSnk.

\medskip
\begin{center}
\fbox{\parbox{0.95\linewidth}{\noindent
{\MINIDSk}
\\[.8ex]
\begin{tabular*}{0.9\textwidth}{rl}
\emph{Instance:} & A bipartite graph $G = (L,R;E)$ and an integer $k$.\\
\emph{Parameter:} & $k$.\\
\emph{Question:} & Does $G$ have an identifiable $\ell$-subgraph induced by a set $J$ with $|J|\le k$?
\end{tabular*}
}}
\end{center}
% \medskip
\begin{center}
\fbox{\parbox{0.95\linewidth}{\noindent
{\MINIDSnk}
\\[.8ex]
\begin{tabular*}{0.9\textwidth}{rl}
\emph{Instance:} & A bipartite graph $G = (L,R;E)$ and an integer $k$.\\
\emph{Parameter:} & $|L|-k$.\\
\emph{Question:} & Does $G$ have an identifiable $\ell$-subgraph induced by a set $J$ with $|J|\le k$?
\end{tabular*}
}}
\end{center}
\medskip
The two problems differ only in the choice of parameter; the \MINIDSnk problem can be reformulated as the problem of finding a set $L_0\subseteq L$ of size at least $k$ such that $G-N[L_0]$ is identifiable. We show that both parameterizations are \W{1}-hard, i.e., they are not fixed-parameter tractable unless $\FPT=\W{1}$, which is deemed unlikely. For both problems we give parameterized reductions from the well-known \W{1}-hard \MCCk problem, defined as follows.
\medskip
\begin{center}
\fbox{\parbox{0.95\linewidth}{\noindent
{\MCCk}
\\[.8ex]
\begin{tabular*}{.9\textwidth}{rl}
\emph{Instance:} & A graph $G = (V,E)$, an integer $k$, and a function $\phi\colon V\to\{1,\ldots,k\}$.\\
\emph{Parameter:} & $k$.\\
\emph{Question:} & Does $G$ contain a clique $C$ with $\phi(C)=\{1,\ldots,k\}$?
\end{tabular*}
}}
\end{center}
\medskip

\begin{theorem}
\MINIDSk is \W{1}-hard.
\end{theorem}

\begin{proof}
We give a parameterized reduction from \MCCk to \MINIDSk. Let $(G=(V,E),\phi,k)$ be an instance of \MCCk.
Without loss of generality assume that $k\geq 3$ or else solve the instance in polynomial time (finding a clique of size $k\in\{1,2\}$).
Let $V_i:=\phi^{-1}(i)$ for $i\in\{1,\ldots,k\}$. We will construct a bipartite graph $G'=(L,R;E')$ such that $(G,k,\phi)$ is yes for \MCCk if and only if $(G',k')$ is yes for \MINIDSk.

\medskip

\emph{Construction.}
We create an instance of \MINIDSk with bipartite graph $G'=(L,R;E')$ and parameter $k'=2k$.
The set $L$ consists of the vertices in $V$ along with $k$ special vertices $t_1,\ldots,t_k$. We now describe the set $R$ along with the adjacencies between $L$ and $R$:
\begin{itemize}
 \item For each choice of $1\leq i<j\leq k$ we create a set $E_{ij}$ of vertices, which is then added to $R$.
  \begin{itemize}
   \item For each edge $\{u,v\}$ in $G$ with $u\in V_i$ and $v\in V_j$ we add a vertex to $E_{ij}$ and make it adjacent to $u$ and $v$ in $G'$. (We could also achieve this by starting with $G$, dropping the (irrelevant) edges between vertices of the same set $V_i$, and then subdividing every edge.)
   \item Make the special vertex $t_i$ adjacent to all vertices of $E_{ab}$ with $1\leq a<b\leq k$ and $i\notin \{a,b\}$.
  \end{itemize}
 \item For each $1\leq i\leq k$ create a set $F_i$ of $k\cdot |V_i|$ vertices and add it to $R$.
 \begin{itemize}
  \item Make each vertex $v\in V_i$ adjacent to $k$ private vertices in $F_i$. No other vertices of $V$ will be adjacent to these vertices.
  \item Make each special vertex $t_j$ adjacent to all vertices of $F_i$ with $i\neq j$.
 \end{itemize}
\end{itemize}
This completes the construction of $G'=(L,R;E')$. It can be helpful to keep in mind that vertices in $V_i$ are only adjacent to (some) vertices in $F_i$ or in $E_{ab}$ with $i\in\{a,b\}$, whereas each special vertex $t_i$ is adjacent to all vertices in $F_j$ with $i\neq j$ and all vertices in $E_{ab}$ with $i\notin\{a,b\}$. An example construction is shown in Fig.~\ref{fig:reduction}.

\begin{figure}[h!]
  \centering
   \includegraphics[width=160mm]{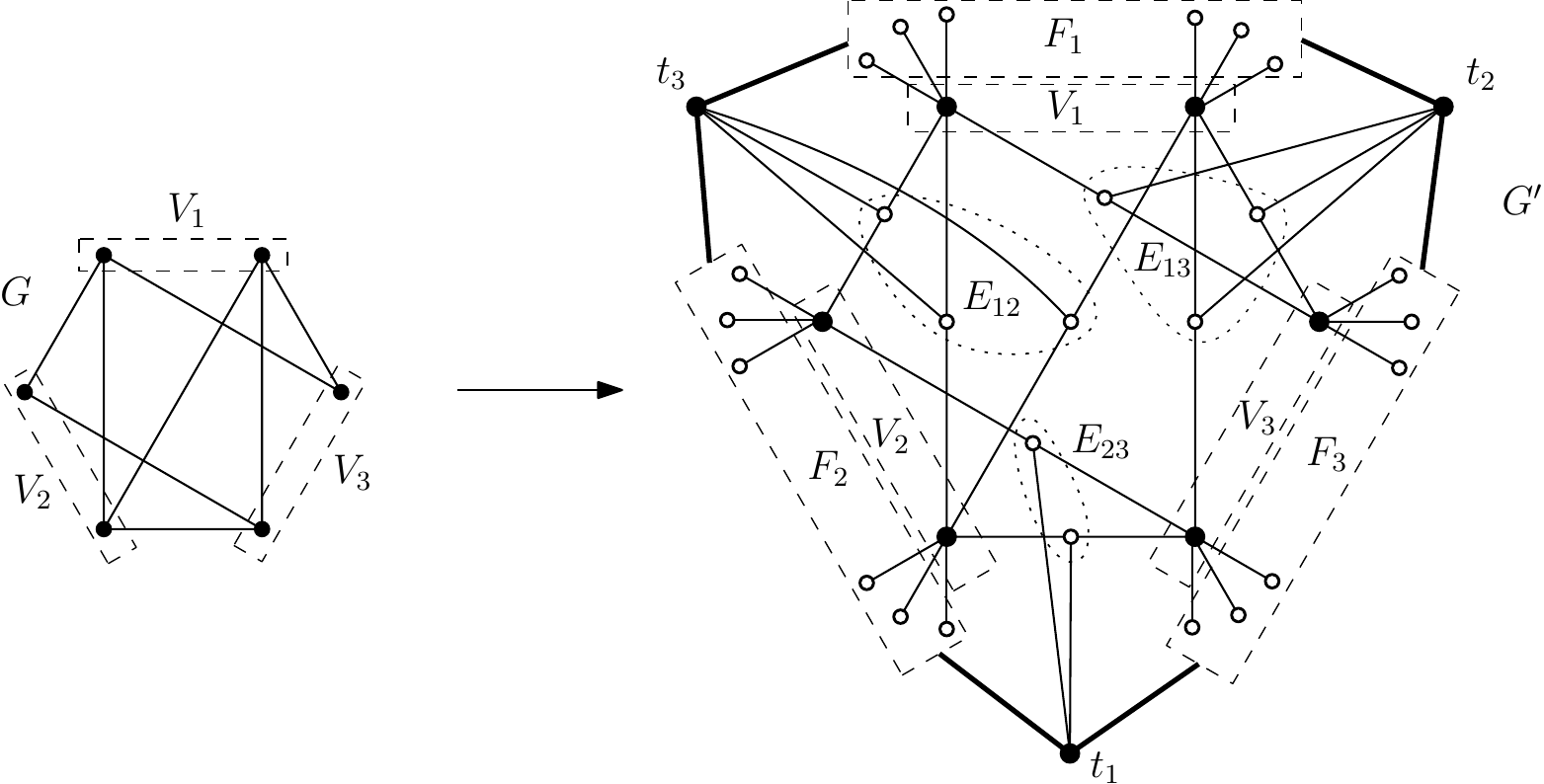}
\caption{An example construction of the bipartite graph $G' = (L,R;E')$ from an input $(G,k,\phi)$ to \MCCk.
  In the example $k = 3$, vertices in $L$ and $R$ are colored black and white, respectively, and a thick edge between vertex $t_i$ and a set $F_j$ means that $t_i$ is adjacent to all vertices in $F_j$.}\label{fig:reduction}
\end{figure}

Clearly, the construction can be performed in polynomial time. It remains to prove correctness, that is, that
$G$ has a $k$-clique containing exactly one vertex of each set $V_i$ if and only if $(G',k')$ is a yes instance of \MINIDSk for $k'=2k$.

\medskip

\emph{Correctness.}
Assume first that $G$ contains a $k$-clique $C$ with exactly one vertex from each set $V_i$ and let $\{v_i\}=C\cap V_i$. We claim that $L'=\{v_1,\ldots,v_k,t_1,\ldots,t_k\}$ induces an identifiable $\ell$-subgraph. Let $R'\subseteq R$ denote the vertices in $R$ of the induced $\ell$-subgraph, i.e., the vertices that have no neighbor among $L\setminus L'$.

Let us first check that there is a matching of $L'\setminus\{t_i\}$ into $R'\setminus N(t_i)$: The vertex $t_i$ is not adjacent to $F_i$ nor to sets $E_{ab}$ with $i\notin\{a,b\}$. There are $k$ vertices in $F_i$ that are adjacent to $t_1,\ldots,t_{i-1},t_{i+1},\ldots,t_k$ as well as the vertex $v_i$. These are contained in $R'$ since all their neighbors are in $L'$. We can match the mentioned vertices of $L'$ to them. Because $C$ is a clique there are edges from $v_i$ to each other vertex of the clique; these give rise to vertices in $E_{1i},\ldots,E_{i-1,i}, E_{i,i+1}, \ldots,E_{i,k}$ corresponding to these edges that have no other neighbors in $V\subseteq L$ (in $G'$). Because $t_i$ is not adjacent to such sets $E_{ab}$ all these vertices are present and each $v_1,\ldots,v_{i-1},v_{i+1},\ldots,v_k$ can be matched to the vertex representing its edge to $v_i$.

Let us now check that there is a matching of $L'\setminus\{v_i\}$ into $R'\setminus N(v_i)$: For each $v_j$ with $j\neq i$ all its neighbors in $F_j$ are present since they have no other neighbor in $V$ and all $t_1,\ldots,t_k$ are in the $\ell$-subgraph. Thus, each $v_j$ can be matched to such a neighbor. Because $k\geq 3$ there are at least two sets $F_j$ and $F_{j'}$ with $i\notin \{j,j'\}$ and a total of $2k-2$ vertices therein are not yet matched to. Thus, all vertices $t_1,\ldots,t_k$ can be matched to these vertices. (E.g., all but $t_j$ to vertices of $F_j$ and $t_j$ to a vertex of $F_{j'}$.)

Thus, the $\ell$-subgraph induced by $L'$ is indeed identifiable. Since $|L'|=2k$ this implies that $(G',2k)$
is a yes instance of \MINIDSk.

Now assume that $(G',2k)$ is yes for \MINIDSk. Let $L'\subseteq L$ be a set of size at most $2k$ such that the $\ell$-subgraph of $G'$ induced by $L'$, namely $G'[L',R']$ with $R'=R\setminus N(L\setminus L')$, is identifiable. Our goal is to show that it includes all special vertices $t_1,\ldots,t_k$ along with one vertex per set $V_i$ and that the latter vertices form a $k$-clique in $G$.

We first observe that $L'\cap V\neq\emptyset$ and $L'\cap \{t_1,\ldots,t_k\}\neq \emptyset$ is required: Excluding either type of vertex implies $R'=\emptyset$ since each vertex in $R$ is adjacent to at least one vertex of $V$ and at least one vertex $t_i$.

Assume for contradiction that at least two special vertices, say $t_i$ and $t_j$ with $i\neq j$, are not in $L'$. Thus, picking a third vertex $t_\ell\notin \{t_i,t_j\}$ we need a matching of $L'\setminus \{t_i,t_k,t_\ell\}$ into the vertices of $R$ that are not neighbors of (at least) $t_i$, $t_j$, and $t_\ell$, but no such vertices exist: There is no $F_a$ with $a= i$ and $a=j$ and there is no $E_{ab}$ with $i\in\{a,b\}$, $j\in\{a,b\}$, $\ell\in\{a,b\}$. Since we have at least one vertex $v\in L'\cap V$ we can observe that this vertex cannot be matched; a contradiction.

Now assume for contradiction that exactly one special vertex, say $t_1$, is not contained in $L'$. This requires a matching of $L'\setminus\{t_2\}$ into $R'\setminus N(t_2)$. Since both $t_1,t_2\notin L'\setminus \{t_2\}$, no vertex of a set $E_{ab}$ may exist in $R'\setminus N(t_2)$ except possibly for vertices of $E_{12}$. Thus, no vertex of $V_3,\ldots,V_k$ may be in $L'$ since they would have no neighbors to match to. This in turn implies that no other set $E_{ab}$ except for $E_{12}$ has any vertices in $R'$. We complete the contradiction by considering the requirement of a matching of $L'\setminus \{t_3\}$ into $R'\setminus N(t_3)$: Now, the vertices of $E_{12}$ are not available since they are adjacent to $t_3$. Thus, there are no $E_{ab}$ vertices to match to. Similarly, absence of $t_1$ and $t_3$ eliminates all vertices of sets $F_a$. Since $L'\setminus\{t_3\}$ must contain at least one vertex of $V$, we find that such a vertex cannot be matched into $R'\setminus N(t_3)$; a contradiction.

We now have the remaining case that $\{t_1,\ldots,t_k\}\subseteq L'$.
We also know already that at least one vertex of $V$ must be contained in $L'$, say $V_1\cap L'\neq \emptyset$ and pick $v_1\in V_1\cap L'$. Assume for contradiction that some set $V_i$ with $i\neq 1$ has an empty intersection with $L'$. It follows that in $R'$ there are no vertices of sets $E_{ab}$ with $i\in\{a,b\}$ since each such vertex is adjacent to some vertex in $V_i$. We now consider the requirement of a matching of $L'\setminus\{t_i\}$ into $R'\setminus N(t_i)$ to complete the contradiction: This additionally ensures that there are no vertices of $F_1$ left in $R'\setminus N(t_i)$ as well as no vertices of $E_{ab}$ with $i\not\in \{a,b\}$, implying that there are no neighbors for $v_1$ to match to; a contradiction.

Thus, we have $\{t_1,\ldots,t_k\}\subseteq L'$ and $L'$ has a nonempty intersection with each set $V_1,\ldots,V_k$. Because $L'$ has size at most $2k$ this directly implies that its size is exactly $2k$ and that it contains exactly one vertex of each set $V_i$, say $\{v_i\}=L'\cap V_i$. It remains to show that $\{v_1,\ldots,v_k\}$ is a clique in $G$. Assume for contradiction that this is not the case, say that $v_i$ and $v_j$ for $i\neq j$ are not adjacent in $G$. Consider the requirement of a matching of $L'\setminus \{t_i\}$ into $R'\setminus N(t_i)$: Absence of $t_i$ ensures that no vertex of $F_j$ is present. Moreover, no vertices of $E_{ab}$ for $i\notin \{a,b\}$ are in $R'\setminus N(t_i)$. In particular, for vertex $v_j$ this only leaves vertices in $E_{ij}$ or $E_{ji}$ (depending on whether $i<j$ or $i>j$). Because $v_i$ and $v_j$ are not adjacent, however, and no other vertex of $V_i$ is in $L'$, no such vertices exist in $R'\supseteq R'\setminus N(t_i)$. Thus, $v_j$ cannot be matched; a contradiction.

It follows that the vertices $v_1,\ldots,v_k$ must indeed form a clique in $G$. This completes the proof.
\end{proof}

The above proof also shows that the problem of testing whether a given bipartite graph \hbox{$G = (L,R;E)$} has an identifiable $\ell$-subgraph induced by a set $J\subseteq L$ with $|J|= k$ is \W{1}-hard (with respect to parameter $k$).

\begin{theorem}\label{theorem:minidsnk}
\MINIDSnk is \W{1}-hard.
\end{theorem}

\begin{proof}
We give a parameterized reduction from \MCCk to \MINIDSnk. Let $(G=(V,E),\phi,k)$ be an instance of \MCCk.
W.l.o.g., assume $k\geq 3$ or else solve the instance in polynomial time. Let $n:=|V|$ and let $V_i:=\phi^{-1}(i)$ for $i\in\{1,\ldots,k\}$. Assume w.l.o.g.~that each set $V_i$ contains at least two vertices (else we can restrict the graph to the subgraph induced by the neighborhood of $v_i\in V_i$ and drop color $i$ to get an equivalent instance).

We will construct a bipartite graph $G'=(L,R;E')$ such that $(G,k,\phi)$ is yes for \MCCk if and only if $(G',k')$ is yes for \MINIDSnk for $k'=|L|-k$. That is, $(G,k,\phi)$ should be yes if and only if the graph $G'$ contains an $\ell$-identifiable subgraph that is induced by a set $L'$ of size at most $|L|-k$. Note that $(G',k')$ has a parameter value of $|L|-k'=|L|-(|L|-k)=k$. Recall that we can equivalently ask for the existence of a set $L_0\subseteq L$ of size at least $k$ such that $G'-N[L_0]$ is identifiable since then $L\setminus L_0$ is of size at most $|L|-k$ and can play the role of the requested set $L'$ (and conversely $L\setminus L'$ is a feasible choice for $L_0$). Define $r:=n+k$; this value will be used in the construction.

\medskip

\emph{Construction.}
The graph $G'=(L,R;E')$ is defined as follows:
\begin{itemize}
 \item The vertex set $L$ consists of $V$ as well as a set $T$ of special vertices $t_1,\ldots,t_k$.
 \item The set $R$ contains for each vertex $v\in V$ a set of $k+1$ vertices $p_{v,1},\ldots,p_{v,k+1}$ whose only neighbor in $V$ will be $v$ (so they are in a limited sense a private neighbors of $v$). Let $F_i$ denote the set of vertices $p_{v,\ell}$ with $v\in V_i$ for each $i\in\{1,\ldots,k\}$. (The exact number of these vertices per vertex $v$ will be immaterial so long as they are at least $k+1$.)
 \item The set $R$ furthermore contains vertices derived from the edges of $G$. Let $e=\{v_i,v_j\}$ be any edge of $G$ with $v_i\in V_i$ and $v_j\in V_j$. Create $r$ vertices $q_{e,1},\ldots,q_{e,r}$ and add them to $R$. Make each of them adjacent to all vertices of $V_i\setminus \{v_i\}$ and all vertices of $V_j\setminus\{v_j\}$. Do this for all edges for any $1\leq i<j\leq k$ and let $E_{ij}$ contain the vertices $q_{e,1},\ldots,q_{e,r}$ for edges $e$ between $V_i$ and $V_j$ in $G$. The set $R$ is thus the union of sets $F_a$ for $1\leq a\leq k$ and sets $E_{ab}$ for $1\leq a<b\leq k$. (Again, the exact value of $r$ is not important so long as $r\geq n+k$.)
 \item Make each vertex $t_i\in T\subseteq L$ adjacent to all vertices of each set $F_a$ with $1\leq a\leq k$. Furthermore, make each $t_i$ adjacent to all vertices of each set $E_{ab}$ with $i\notin\{a,b\}$.
\end{itemize}
Define $k'=|L|-k$ and return the instance $(G',k')$. An example construction is shown in Fig.~\ref{fig:reduction-2}.

\begin{figure}[h!]
  \centering
   \includegraphics[width=170mm]{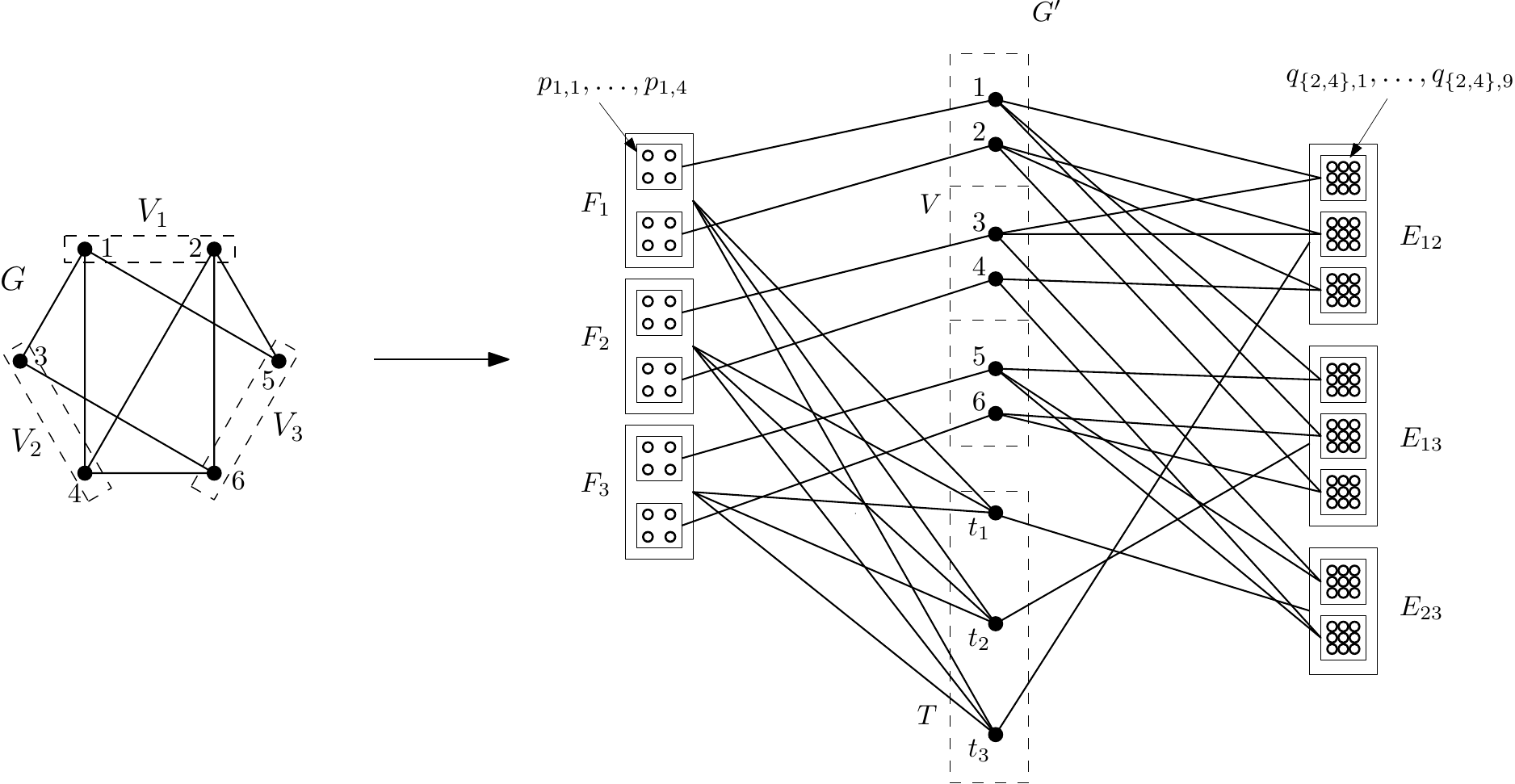}
  \caption{An example construction of the bipartite graph $G' = (L,R;E')$ from an input $(G,k,\phi)$ to \MCCk.
  In the example $k = 3$, vertices in $L$ and $R$ are colored black and white, respectively, and an edge between
  vertex $u\in L$ and a set $S$ of vertices in $R$ means that $u$ is adjacent to all vertices in $S$.}\label{fig:reduction-2}
\end{figure}

Clearly this construction can be performed in polynomial time and we already pointed out that the parameter value of $(G',k')$ is equal to $k$.
(Parameter value bounded by any function of $k$ would be enough for a parameterized reduction.)
It remains to prove correctness, that is, that $G$ has a $k$-clique containing exactly one vertex of each set $V_i$
if and only if $(G',k)$ is a yes instance of \MINIDSnk.

\medskip

\emph{Correctness.}
Assume first that $(G,k,\phi)$ is yes for \MCCk. Thus, $G$ contains a clique $C$ containing exactly one vertex of each set $V_i$.
We claim that $L':=L\setminus C$ induces an identifiable $\ell$-subgraph in $G'$. The $\ell$-subgraph induced by $L'$ is exactly $G''=G'[L',R']$ where $R'=R\setminus N(L\setminus L')=R\setminus N(C)$. We need to show that $G'$ is identifiable.

Let us first see that $G''$ has a matching of $L'\setminus \{w\}$ into $R'\setminus N(w)$ for any $w\in V\cap L'$. Fix any such vertex $w$: Recall that in the construction we made for each vertex $v\in V$ vertices $p_{v,1},\ldots,p_{v,k+1}$ such that $v$ is their only neighbor in $V$. For each vertex $v\in V\setminus (L'\cup\{w\})$ all these vertices $p_{v,\ell}$ exist in $R'\setminus N(w)=R\setminus N(C\cup\{w\})$ since $C\cup\{w\}\subseteq V$ so we can match each $v$ to $p_{v,1}$. Moreover, for any $v\in L'\setminus \{w\}$ we can match all $k$ vertices of $T$ to $p_{v,2},\ldots,p_{v,k+1}$. (Here we tacitly assume that $G$ has more than $k+1$ vertices, which is w.l.o.g.) This completes the required matching.

Let us now exhibit a matching in $G''$ of $L'\setminus \{t_i\}$ into $R''=R'\setminus N(t_i)$ for any $t_i\in T\subseteq L'$: Note that $R''$ in particular does not contain vertices of sets $F_a$ for $1\leq a\leq k$ nor vertices of $E_{ab}$ for $i\notin \{a,b\}$ since all those are adjacent to $t_i$ in $G'$ (so they are also adjacent to $t_i$ in $G''$). It remains to use vertices of $E_{ab}$ with $i\in\{a,b\}$, recalling that many of them are not present already in $R'=R\setminus N(C)$. Fix $j\neq i$. Let $\{v_i\}:= C\cap V_i$ and $\{v_j\}:= C\cap V_j$. Since $C$ is a clique, vertices $v_i$ and $v_j$ are adjacent in $G$. For the corresponding edge $e=\{v_i,v_j\}$ we created vertices $q_{e,1},\ldots,q_{e,r}$ in $E_{ij}\subseteq R$ (or $E_{ji}$ if $j<i$). Each $q_{e,l}$ is adjacent to all of $V_i\setminus \{v_i\}$ and $V_j\setminus\{v_j\}$ but not to $v_i$ or $v_j$; they are not adjacent to any vertex of $V\setminus(V_i\cup V_j)$. Thus, all of these vertices are present in $R'$ and hence also in $R''=R'\setminus N(t_i)$. We can therefore match all vertices of $V_i\setminus \{v_i\}$, $V_j\setminus\{v_j\}$, and $T\setminus\{t_i,t_j\}$ to them since these are in total less than $r=n+k$ vertices. By repeating the argument for all $j'\in\{1,\ldots,k\}\setminus\{i,j\}$ we can also match vertices in $\bigcup_{j'}\left(V_{j'}\setminus C\right)\cup\{t_j\}$, obtaining a matching for all of $L'\setminus\{t_i\}=L\setminus (C\cup \{t_i\})$. (Here we need that $k\ge 3$ so that we can match $t_j$.)

It follows that the $\ell$-subgraph induced by $L'$ is indeed identifiable. Since $|L'|\leq |L|-k$ it follows that $(G',k')$ is yes for \MINIDSnk.

Assume now that $(G',k')$ is yes for \MINIDSnk and let $L'$ be a subset of $L$ of size at most $k'=|L|-k$ such that the $\ell$-subgraph induced by $L'$, namely $G'':=G'-N[L']=G'[L',R']$ where $R'=R\setminus N(L\setminus L')$, is identifiable. Let $C:=L\setminus L'$. This is a set of size at least $k$. We will show that $C$ is a subset of $k$ vertices of $V$ that form a clique in $G$ with $|C\cap V_i|=1$ for $1\leq i\leq k$. Note that $R'=R\setminus N(C)$.

We begin with some observations: If $V\subseteq C$ then $R'=\emptyset$ because each vertex of $R$ is adjacent to at least one vertex of $V$. In this case, $G''$ could not be identifiable since that requires having at least one edge. Thus, $V\setminus C\neq \emptyset$ and we pick an arbitrary vertex $v_0\in V\setminus C\subseteq L'$ to be used later. Similarly, if $T\subseteq C$ then again $R'=\emptyset$ and $G'[L',R']$ cannot be identifiable. We pick $t_0\in T\setminus C$ to be used later. (Note that $t_0=t_i$ for some $1\leq i\leq k$.)

We will now prove several restrictions on $C$ by contradiction-based arguments. The first two aim at proving that $T\cap C=\emptyset$.

Assume for contradiction that $|T\cap C|\geq 2$. Thus, $G''$ being identifiable implies that there must be a matching of $L'\setminus \{t_0\}$ to $R'\setminus N(t_0)$. Say $t_i,t_j\in T\cap C$ with $i\neq j$, then $R'\setminus N(t_0)$ contains no vertices of $R$ that are adjacent to any of $t_0$, $t_i$, or $t_j$. This is a contradiction since every vertex of $R$ is adjacent to at least one of them: Vertices in any $F_a$ are adjacent to each vertex of $T$ and vertices in any $E_{ab}$ are only not adjacent to two vertices of $T$, namely $t_a$ and $t_b$. Thus, $|T\cap C|\leq 1$.

Assume for contradiction that $|T\cap C|=1$ and assume w.l.o.g.~that $\{t_1\}=T\cap C$. Because $G''$ is identifiable there must be a matching of $L'\setminus\{t_2\}$ into $R'\setminus N(t_2)$ in $G''$. Note that in $R'\setminus N(t_2)$ there is no vertex of any set $F_a$. Similarly, vertices of $E_{ab}$ are not in $R'\setminus N(t_2)$ unless $1\in\{a,b\}$ and $2\in\{a,b\}$ since otherwise they are adjacent to at least one of $t_1$ or $t_2$. This in turn implies that $L'$ contains no vertices from $V_a$ for $a\notin\{1,2\}$ since they have no neighbors in $R'\setminus N(t_2)$. Consequently, even $R'\supseteq R'\setminus N(t_2)$ contains no vertex of any set $E_{ab}$ with $(a,b)\neq(1,2)$ because the vertices of other sets $E_{ab}$ are all adjacent to some vertex of $V\setminus(V_1\cup V_2)\subseteq L\setminus L'=C$. Now, consider the requirement of a matching of $L'\setminus\{t_3\}$ into $R'\setminus N(t_3)$ in $G''$. We now get that there are no vertices of sets $F_a$ nor of sets $E_{ab}$. The latter holds because in $R'$ only vertices of $E_{12}$ can exist but all those are adjacent to $t_3$ and hence not in $R'\setminus N(t_3)$. Thus, $R'\setminus N(t_3)$ is empty and we cannot match the vertex $v_0$ anywhere; a contradiction.

We now know that $C\cap T=\emptyset$. Assume that $C$ contains at least two vertices of the same set $V_i$, i.e., that $|C\cap V_i|\geq 2$. Let $v_i,v'_i\in C\cap V_i$ with $v_i\neq v'_i$. Crucially, for any $1\leq a<b\leq k$ with $i\in\{a,b\}$, all vertices in $E_{ab}$ are adjacent to at least one of $v_i$ and $v'_i$, by construction: A vertex $q_{e,l}$ for $e=\{p,q\}$ with $p\in V_i$ and $q\in V_j$ is adjacent to all vertices of $V_i\setminus \{p\}$. Because only one of $v_i$ and $v'_i$ can be equal to $p$ it follows that $q_{e,l}$ is adjacent to at least one of the two. Thus, no vertex of $E_{a,b}$ is present in $R'\subseteq R\setminus N(\{v_i,v'_i\})$ for $a$ and $b$ with $i\in\{a,b\}$.

Now, because $G''$ must be identifiable there must be a matching of $L'\setminus\{t_i\}$ into $R''=R'\setminus N(t_i)$. There are no vertices of any set $F_a$ in $R''$ because all of them are adjacent to $t_i$. For the same reason there are no vertices of sets $E_{ab}$ when $i\notin \{a,b\}$. From above we know that for $E_{ab}$ with $i\in\{a,b\}$ no such vertices are present in $R'\supseteq R''$. Thus, vertex
$v_0\in V\setminus C\subseteq L'$ cannot be matched; a contradiction.

At this point we know that $C$ contains no vertices of $T$ and at most one vertex of each set $V_i$. Because the size of $C$ is at least $k$ this implies that $C$ contains exactly one vertex of each set $V_i$ and no further vertices (it is of size exactly $k$). Let $\{v_i\}=C\cap V_i$ for $1\leq i\leq k$. It remains to show that the vertices $v_1,\ldots,v_k$ form a clique in $G$.

Assume for contradiction that $v_i\in V_i$ and $v_j\in V_j$ are not adjacent in $G$ for some $i\neq j$; w.l.o.g.~$i<j$. Again the construction of the edge-related vertices $q_{e,l}$ is important here: Consider any edge $e=\{v'_i,v'_j\}$ with $v'_i\in V_i$ and $v'_j\in V_j$. s(Here we also tacitly assume that there is such an edge, which is w.l.o.g.~as otherwise $(G,k,\phi)$ is a no instance to \MCCk.)
Because there is no edge between $v_i$ and $v_j$, we must have $v'_i\neq v_i$ or $v'_j\neq v_j$ (or both). This directly implies that each vertex $q_{e,l}\in E_{ij}$ is adjacent to $v_i$ or $v_j$ (or both) because it is adjacent to all of $V_i\setminus\{v'_i\}$ and all of $V_j\setminus\{v'_j\}$. It follows that there are no vertices of $E_{ij}$ in $R'\subseteq R\setminus N(\{v_i,v_j\})$.

Because $G''$ is identifiable there must be a matching of $L'\setminus\{t_i\}$ into $R'\setminus N(t_i)$. In the latter set there are no vertices of any set $F_a$ because they are subsets of $N(t_i)$ and similarly no vertices of $E_{ab}$ if $i\notin\{a,b\}$. Consider any vertex $v''_j\in V_j\setminus \{v_j\}\subseteq L'$ (using that $C\cap V_j=\{v_j\}$ and $|V_j|\geq 2$): Its neighbors in
$G'$ are in $F_j$ and in sets $E_{ab}$ with $j\in\{a,b\}$. This leaves only the set $E_{ij}$ since we need $i\in\{a,b\}$ and $j\in\{a,b\}$ or else none of the vertices are in $R'$, but we already know that no vertices of $E_{ij}$ are in $R'\supseteq R'\setminus N(t_j)$; a contradiction.

It follows that the vertices of $C$ form a clique in $G$, with exactly one vertex from each set $V_i$; this proves that $(G,k,\phi)$ is yes for \MCCk and completes the proof.
\end{proof}

%%%%%%%%%%%%%%%%%%%%%%%%%%%%%%%%%%%%%%%%%%%%%%%%%%%%%%%%%%%%%%%

\section{Concluding remarks}\label{sec:conclusion}

\begin{sloppypar}
In this paper, we showed that the \IDS and the \MAXIDS problems are polynomially solvable and that
two natural parameterized variants of the \MINIDS problem are \W{1}-hard. Regarding approximation issues, the \MINIDS problem was shown to be \APX-hard~\cite{DAM1}, however, its exact (in)approximability status remains an open question.

In~\cite{DAM1}, two other \NP-hard problems related to identifiability were studied: finding the minimum number of edges that one must delete from a given identifiable graph to destroy identifiability and finding the smallest size of a set $R'\subseteq R$ such that the graph $G[L,R']$ is identifiable.
The hardness proof for the former problem shows that the problem is also \W{1}-hard with respect to its natural parameterization. More precisely,
one can combine the \NP-hardness proof from~\cite{DAM1} and the proof of the \NP-hardness of the problem from which that reduction was made (finding the minimum number of edges that one must delete from a given bipartite graph in order to decrease its
matching number)~\cite[Theorems 3.2 and 3.3 and their proofs]{MR2519166} to obtain a parameter-preserving reduction from the parameterized clique problem, with respect to its natural parameterization, which is \W{1}-hard. We leave for future research the determination of the parameterized complexity status of the latter problem, as well as the (in)approximability status of both problems.
\end{sloppypar}

\section*{Acknowledgements}

This work was supported in part by the Slovenian Research Agency (I$0$-$0035$, research program P$1$-$0285$ and research projects N$1$-$0032$, J$1$-$5433$, J$1$-$6720$, J$1$-$6743$, and J$1$-$7051$). Part of this research was carried out during the visit of M.M. to S.K. at University of Bonn; their hospitality and support is gratefully acknowledged.

\end{document}